\documentclass[12pt]{article}
\usepackage{amssymb,amsmath,amsthm,amscd,latexsym}

\usepackage{amsfonts}
\usepackage[mathscr]{eucal}
\usepackage{color}

\newtheorem{thm}{Theorem}[section]

\newtheorem{lem}[thm]{Lemma}
\newtheorem{cor}[thm]{Corollary}
\newtheorem{exam}{Example}

\font\msbm=msbm10 at 12pt
\newcommand{\ZZ}{\mbox{\msbm Z}}

\newcommand{\FF}{\mbox{\msbm F}}

\newcommand{\G}{\alpha}

\def\vv{\mathbf{v}}

\def\vw{\mathbf{w}}
\def\vc{\mathbf{c}}

\def\1v{\mathbf{1}}
\def\0v{\mathbf{0}}

\begin{document}
\title{ Constructions of Self-Dual and Formally Self-Dual Codes from Group Rings}
\author{
Steven T. Dougherty \\
Department of Mathematics \\
University of Scranton \\
Scranton, PA 18510 \\
USA \\
Joseph Gildea \\
Rhian Taylor \\ 
University of Chester \\
Chester, UK \\
Alexander Tylshchak \\
Department of Algebra \\ 
Uzhgorod State University \\
Ukraine
}

\maketitle

\begin{abstract}
We give constructions of self-dual and formally self-dual codes from group rings where the ring is a finite commutative Frobenius ring.
We improve the existing construction given in \cite{Hurley1} by showing that one of the conditions given in the theorem is unnecessary and moreover it restricts the number of self-dual codes obtained by the construction.  We show that several of the standard constructions of self-dual codes are found within our general framework.  We prove that our constructed codes correspond to ideals in the group ring $RG$ and as such must have an automorphism group that contains $G$ as a subgroup.  We also prove that a common construction technique for producing self-dual codes cannot produce the putative $[72,36,16]$ Type~II code. Additionally, we show precisely which groups can be used to construct the extremal Type II codes over length 24 and 48. 
\end{abstract} 


\section{Introduction}

Self-dual codes over fields and rings are one of the most important and widely studied families of codes.  They have interesting connections to groups, designs, lattices and other objects as well.  As such, constructions of interesting self-dual codes are an important area of study in coding theory.
In \cite{Hurley1}, Hurley gave a construction of self-dual codes from elements in a group algebra.  The constructions were done generally in the group algebra $\FF_2D_{2k},$ where $D_{2k}$ is the dihedral group of order $2k.$    In \cite{McLoughlin}, McLoughlin gave a construction of the extremal $[48,24,12]$  using this construction technique.    
In this paper, we expand this construction to codes over finite commutative Frobenius rings and show how to construct isodual and  formally self-dual codes as well.  Additionally, we construct self-dual and formally self-dual codes over various families of rings, which, in turn,  give formally self-dual and self-dual binary codes via a Gray map.  We consider additional groups as well and expand the constructions using these groups.
 
\subsection{Codes}  
 
Let $R$ be a finite ring.  We assume that all rings contain a multiplicative identity.   
Let $\widehat{R}$ be the character module of $R$.  Then for a finite ring $R$ the following are equivalent.
\begin{itemize}
\item $R$ is a Frobenius ring.
\item As a left module, $\widehat{R} \cong  { }_RR.$
\item As a right module, $\widehat{R} \cong R_R.$ 
\end{itemize}
For commutative rings we can say that the $R$-module $R$ is injective and that if $R$ is a finite local ring with maximal ideal $\mathfrak m$ and residue field $\mathbf k$, then a Frobenius ring has 
$\mathrm{dim}_{\mathbf k} \mathrm{Ann}(\mathfrak m)=1$.
Throughout this paper, we shall always assume that the rings are commutative.

 A code over $R$ of length $n$ is a subset of $R^n$.  If the code is a submodule of the ring then we say that the code is a linear code.  We attach to the ambient space the usual inner-product, namely $[\vv,\vw]= \sum v_i w_i$ and define the orthogonal with respect  to this inner-product as $C^\perp = \{ \vv  \in R^n\ | \ [\vv,\vw]=0,\ \forall \vw \in C \}.$  There is a unique orthogonal because the ring is commutative.    A code is said to be self-orthogonal if $C \subseteq C^\perp$ and  self-dual if $C=C^\perp$.  We say that two codes $C$ and $C'$ are equivalent if $C'$ can be formed from $C$ by permuting the coordinates of $C$.  Note that we are not allowing for multiplication of a coordinate by a unit in our definition of equivalent.  A code $C$ is said to be isodual if $C$ and $C^\perp$ are equivalent codes. The automorphism group of a code $C$, denoted $Aut(G)$, consists of all permutations of the coordinates of the code that fix the code.
 
 Let $C$ be a code  over a ring $R= \{a_0,a_1,\dots,a_{r-1}\}$.  The complete weight enumerator for the code $C$  is defined as:
\begin{equation} \label{equation:cwe} 
cwe_C(x_{a_0},x_{a_1}, \dots, x_{a_{r-1}}) = \sum_{\vc \in C} \prod_{i=0}^{r-1} x_{a_i}^{n_i(\vc)}
\end{equation}
where there are $n_i(\vc)$ occurrences of $a_i$ in the vector $\vc$. 

The Hamming weight of a vector $\vv \in R^n$ is $wt_H(\vv) = |\{ i \ | \ v_i \neq 0 \} |.$  
The  Hamming weight enumerator  is given by
\begin{equation} W_C(x,y) = \sum_{\vc \in C}  x^{n-wt_H(\vc) } y^{wt_H(\vc)} = cwe_C(x,y,y,\dots,y). \end{equation}

Throughout this work we restrict ourselves to Frobenius rings since this is the class of rings for which MacWilliams relations exist.  That is, the weight enumerator of a code over a Frobenius ring uniquely determines the weight enumerator of its orthogonal.  The MacWilliams relations imply that for a code $C$ over a Frobenius ring $R$ we have $|C||C^\perp| = |R|^n.$  This often fails for codes over non-Frobenius rings.  In that sense, it is very difficult to discuss self-dual and formally self-dual codes over non-Frobenius rings.

A Gray map is a distance preserving map $\phi$ from $R$ to $\FF_2^t$ for some $t$.  We define the Lee weight, $wt_L(a)$ of an element $a \in R$ as the Hamming weight of $\phi(a).$ We then extend this to $R^n$ by saying that the Lee weight of a vector is the sum of the Lee weights of the coordinates of the vector.   Then the Lee weight enumerator of a code $C$ over $R$ with an associated Gray map is defined as:
\begin{equation}
L_C(x,y) = \sum_{\vc \in C}  x^{N-wt_L(\vc) } y^{wt_L(\vc)},
\end{equation}
where $N$ is the length of the binary image of the code $C$ under the Gray map.
Note that the Lee weight enumerator of a code $C$ is the Hamming weight enumerator of the code $\phi(C).$

We say that a code $C$ is formally self-dual with respect to a weight enumerator if the weight enumerators of  $C$ and $C^\perp$ are identical.
Note that a self-dual code is necessarily formally self-dual with respect to all weight enumerators but a code can be formally self-dual and not self-dual.  Moreover, a code can be formally self-dual with respect to one weight enumerator and not another.

\subsection{Group Rings}

Let $G$ be a finite group or order $n$, then the group ring $RG$ consists of $\sum_{i=1}^n \alpha_i g_i$, $\alpha_i \in R$, $g_i \in G.$  
Addition in the group ring is done by coordinate addition, namely $\sum_{i=1}^n \alpha_i g_i + \sum_{i=1}^n \beta_i g_i = 
\sum_{i=1}^n (\alpha_i + \beta_i ) g_i.$ The product is given by 
$(\sum_{i=1}^n \alpha_i g_i)( \sum_{j=1}^n \beta_j g_j)  = \sum_{i,j} \alpha_i \beta_j g_i g_j .$ This gives that the coefficient of $g_i$ in the product is $ \sum_{g_j g_k = g_i } \alpha_i \beta_j .$  

Group rings are defined for groups and rings of arbitrary cardinality but, in this paper, we shall only be concerned with finite rings and finite groups.  If $R$ is a field then the term group algebra is usually used in this case since  the structure is  an algebra as well.
Throughout this paper we use $e_G$ to refer to the identity element of any group $G$.

We denote the space of $n$ by $n$ matrices with coefficients in $R$ by $M_n(R).$  Note that $M_n(R)$ is, in general, a non-commutative ring. 

A matrix $M$, where the indices are given by the elements in $\ZZ_n$, is said to be circulant if $M_{i,j} = M_{1,j-i \pmod{n} }$, that is the matrix is formed by cycling the first row to the right.  
A matrix $M$, where the indices are given by the elements in $\ZZ_n$, is said to be reverse circulant if 
$M_{i,j} = M_{1,j+i \pmod{n} }$, that is the matrix is formed by cycling the first row to the left.   It is immediate from the definition that a reverse circulant matrix is symmetric, that is $M=M^T.$  

\subsection{Family of Rings}

In this section, we shall describe a family of rings which is useful in producing binary formally self-dual codes via their associated Gray maps.

Define the ring  $R_k$ as \begin{equation}
R_k=\FF_2[u_1,u_2,\ldots, u_k]/\langle {u_i}^2, u_iu_j - u_ju_i \rangle.
\end{equation}

 These rings are  local rings of  characteristic 2 with 
    maximal ideal $\frak{m}=\langle   u_1, u_2, \ldots, u_k \rangle.$
    The socle for the ring $R_k$ is 
    $Soc(R_k)=\langle u_1u_2 \cdots u_k \rangle = \frak{m}^\perp.$  We have that $|R_k| = 2^{2^k}.$
    The rings $R_k$ were described in  \cite{Dougherty5}, 
    \cite{Dougherty6}, and \cite{Dougherty8}.

We can describe a Gray map for $R_k$.  We define
$\phi_1(a+bu_1) = (b,a+b),$ where $\phi$ maps $R$ to $\FF_2^2.$   Then view $R[u_1,u_2,\dots,u_s]$ as $R[u_1,u_2,\dots,u_{s-1}] [u_s]$ and define
$\phi_s(a+bu_s) = (b,a+b).$  Then the map $\phi_k$ is  map from $R_k$ to $\FF_2^{2^k}$.

The following theorem appears in \cite{Dougherty8}.

\begin{thm} \label{theorem:images}
Let $C$ be a self-dual  code over $R_k$ then $\phi_k(C)$ is a self-dual code in $\FF_2^{2^k}.$
\end{thm} 

We shall give several examples where we construct self-dual codes over $R_k$ using the method in the paper and then use the Gray map to construct a binary self-dual code of longer length.

\section{Matrix Construction}

In this section, we shall give a construction of codes in $R^n$ from the group ring $RG.$  This construction was first given for codes over fields by Hurley in \cite{Hurley1}.
Let $R$ be a finite commutative Frobenius ring and let $G = \{ g_1,g_2,\dots,g_n \}$ be a group of order $n$.  Let $v \in RG.$  Define the matrix $\sigma(v) \in M_n(R)$ to be
\begin{equation} \label{equation:construction}
\sigma(v) = 
\left( 
\begin{array}{ccccc}
\alpha_{g_1^{-1} g_1} & \alpha_{g_1^{-1} g_2} & \alpha_{g_1^{-1} g_3} & \dots & \alpha_{g_1^{-1} g_n}  \\
\alpha_{g_2^{-1} g_1} & \alpha_{g_2^{-1} g_2} & \alpha_{g_2^{-1} g_3} & \dots & \alpha_{g_2^{-1} g_n}  \\
\vdots  & \vdots & \vdots & \vdots & \vdots \\ 
\alpha_{g_n^{-1} g_1} & \alpha_{g_n^{-1} g_2} & \alpha_{g_n^{-1} g_3} & \dots & \alpha_{g_n^{-1} g_n}  \end{array}
\right).
\end{equation}

The elements $g_1^{-1}, g_2^{-1}, \dots, g_n^{-1}$ are simply the elements of the group $G$ in some order.  We  take this as the ordering of the elements since it makes the constructions more natural.

For a given element $v \in RG$, we define the following code over the ring $R$:
\begin{equation}
C(v) = \langle \sigma (v) \rangle,
\end{equation}
that is the code formed by taking the row space of $\sigma(v)$ over the ring $R$.
The code $C(v)$ is a linear code since it is the row space of a generator matrix, but it is not possible to determine the size of the code (or the dimension if $R$ is a field) immediately from the matrix.  In other words, the rows of the matrix $\sigma(v)$ are not necessarily linearly independent, although they may be, as we show in the following example.

\begin{exam}
Let $R$ be a finite commutative Frobenius ring and let $G = \{g_1,g_2,\dots,g_n \}$ be a group.  Let $v_1 =  \sum 0g_i$.  Then $\sigma(v_1) $ is the all zero matrix and $C(v_1) = \{ {\bf 0} \}.$  Let $v_2 = \sum \alpha_i g_i$ with $\alpha_j =1$ for some $j$ and $\alpha_i = 0$ for $i \neq j.$   Then $\sigma(v)$ is permutation equivalent to $I_n$, the $n$ by $n$ identity matrix, which gives that $C(v_2) = R^n.$ 
\end{exam}

\begin{exam}  Let $v=(1+s+s^2+s^3)(1+t) \in \mathbb{F}_2M_{16}$ where $M_{16}=\langle s,t\,|\,s^8=t^2=1,\,st=ts^5 \rangle$
is the modular group of order $16$. Then,
\[\sigma(v)=
\left(\begin{smallmatrix}
1&1&1&1&0&0&0&0&1&1&1&1&0&0&0&0\\
0&1&1&1&1&0&0&0&1&0&0&0&0&1&1&1\\
0&0&1&1&1&1&0&0&0&0&1&1&1&1&0&0\\
0&0&0&1&1&1&1&0&1&1&1&0&0&0&0&1\\
0&0&0&0&1&1&1&1&0&0&0&0&1&1&1&1\\
1&0&0&0&0&1&1&1&0&1&1&1&1&0&0&0\\
1&1&0&0&0&0&1&1&1&1&0&0&0&0&1&1\\
1&1&1&0&0&0&0&1&0&0&0&1&1&1&1&0\\
1&1&1&1&0&0&0&0&1&1&1&1&0&0&0&0\\
1&0&0&0&0&1&1&1&0&1&1&1&1&0&0&0\\
0&0&1&1&1&1&0&0&0&0&1&1&1&1&0&0\\
1&1&1&0&0&0&0&1&0&0&0&1&1&1&1&0\\
0&0&0&0&1&1&1&1&0&0&0&0&1&1&1&1\\
0&1&1&1&1&0&0&0&1&0&0&0&0&1&1&1\\
1&1&0&0&0&0&1&1&1&1&0&0&0&0&1&1\\
0&0&0&1&1&1&1&0&1&1&1&0&0&0&0&1
\end{smallmatrix}\right)
\]
\noindent and $\sigma(v)$ is equivalent to
\[
\left(\begin{smallmatrix}
1&0&0&0&0&1&1&1&0&1&1&1&1&0&0&0\\
0&1&0&0&0&1&0&0&1&0&1&1&1&0&1&1\\
0&0&1&0&0&0&1&0&1&1&0&1&1&1&0&1\\
0&0&0&1&0&0&0&1&1&1&1&0&1&1&1&0\\
0&0&0&0&1&1&1&1&0&0&0&0&1&1&1&1
\end{smallmatrix}\right).
\]
Clearly, $C(v)$ is the $[16,5,8]$ Reed-Muller code.
\end{exam} 

We shall now show that the codes we construct are actually ideals in the group ring.  We use this to get information about the automorphism group of the constructed code.

\begin{thm}
Let $R$ be a finite commutative Frobenius ring and $G$ a finite group of order $n$.
Let $v \in RG$ and $C(v) $ the corresponding code in $R^n$.  Let $I(v)$ be the set of elements of $RG$ such that $\sum \alpha_i g_i \in I(V)$ if and only if $(\alpha_1,\alpha_2,\dots, \alpha_n) \in C(v).$  Then $I(v)$ is a left ideal in $RG.$ 
\end{thm}
\begin{proof}
The rows of $\sigma(v)$ consist  precisely of the  vectors that correspond to the elements $hv$ in $RG$ where $h$ is any element of $G$.  The sum of any two elements in $I(v)$ corresponds exactly to the sum of the corresponding elements in $C(v)$ and so $I(v)$ is closed under addition.  

Let $w_1= \sum \beta_i g_i \in RG.$  Then if $w_2$ corresponds to a vector in $C(v),$ it is of the form $\sum \gamma_j h_j v.$  Then
$w_1 w_2 = 
 \sum \beta_i g_i \sum \gamma_i h_i v = \sum \beta_i \gamma_j g_i h_j v$ which corresponds to an element in $C(v)$ and gives that the element is in $I(v).$  Therefore $I(V)$ is a left ideal of $RG.$ 
\end{proof}

\begin{exam}  Let $v=1+ba+ba^2+ba^3 \in \mathbb{F}_2D_8$ where $\langle a,b \rangle \cong D_8$. Then $\sigma(v) = \left(
\begin{smallmatrix}
1&0&0&0&0&1&1&1\\
0&1&0&0&1&1&1&0\\
0&0&1&0&1&1&0&1\\
0&0&0&1&1&0&1&1\\
0&1&1&1&1&0&0&0\\
1&1&1&0&0&1&0&0\\
1&1&0&1&0&0&1&0\\
1&0&1&1&0&0&0&1  
\end{smallmatrix} \right)$ and $\sigma(v)$ is equivalent to 
$A=\left(\begin{smallmatrix}
1&0&0&0&0&1&1&1\\
0&1&0&0&1&1&1&0\\
0&0&1&0&1&1&0&1\\
0&0&0&1&1&0&1&1
\end{smallmatrix} \right)$. Clearly $C(v)=\langle \sigma(v) \rangle$ is the $[8,4,4]$ extended Hamming code. Let $v_1=1+ba+ba^2+ba^3 \in \mathbb{F}_2D_8$,
$v_2=1+b+ba+ba^2 \in \mathbb{F}_2D_8$, $v_3=1+b+ba+ba^3 \in \mathbb{F}_2D_8$ and $v_4=1+b+ba^2+ba^3 \in \mathbb{F}_2D_8$ where $v_i$ are the group ring element
corresponding to the rows of $A$. Let $I(v)=\left\{ \sum_{i=1}^4 \alpha_i v_i | \alpha_i \in \mathbb{F}_2 \right\}$.  Then $I(v)$ is a left ideal of $\mathbb{F}_2D_8$ and
in particular $I(v)$ is the left principle ideal of $\mathbb{F}_2D_8$ generated by $v$.
\end{exam}

\begin{cor} \label{corollary:autogroup} 
Let $R$ be a finite commutative Frobenius ring and $G$ a finite group of order $n$.
Let $v \in RG$ and $C(v) $ the corresponding code in $R^n$. Then the automorphism group of $C(v)$ has a subgroup isomorphic to $G$.  
\end{cor}
\begin{proof}
Since $I(v)$ is an ideal in $RG$ we have that $I(V)$ is invariant by the action of the elements of $G$.  It follows immediately that the automorphism group of $C(v)$ contains $G$ as a subgroup.
\end{proof}

We note that our construction gives a natural generalization of cyclic codes since cyclic codes are ideals in $RC_n$ where $C_n$ is the cyclic group of order $n$.  Cyclic codes are held invariant by the cyclic shift whereas our codes are held invariant by the action of the group $G$ on the coordinates.  Moreover, this is the strength of our construction technique.  Namely,   we can construct a code whose automorphism group must contain a given group.

\begin{exam}
Let $C$ be the extremal $[48,24,12]$ Pless symmetry code.  The automorphism group of this code is $PSL(2,47).$  A computation in GAP \cite{GAP} shows that the only subgroup of $PSL(2,47)$ of order 48 is $D_{48}$.  Hence the only possible construction of this code by our technique must have $G=D_{48}$.  This construction is given by McLoughlin in \cite{McLoughlin}.
\end{exam}

Combining the results in \cite{Bouyuklieva}, \cite{Boyuklieva-2}, \cite{Nebe}, \cite{OBrien} and \cite{Yankov}, we have that the automorphism group of a putative $[72,36,16]$ code must have order 5 or have order dividing 24, see \cite{Open} for details on the automorphism group and a detailed description of this putative code.  Since it is impossible for a group of order 72 to satisfy these we have the following corollary.

\begin{cor} \label{corollary:72} 
The putative $[72,36,16]$ code cannot be of the form $C(v)$ for any $v\in  \FF_2G$ for any group $G.$
\end{cor}
\begin{proof}
The result follows immediately from Corollary~\ref{corollary:autogroup} and the previous discussion. 
\end{proof}

Note that a code whose automorphism group is trivial cannot be constructed by this technique.  For example, in \cite{Lam}, it was shown that if a projective plane of order 10 existed there would be a $[112,56,12]$ self-dual code with no weight 16 vectors that had a trivial automorphism group.  This code was shown not to exist.  Such a code could not be constructed in a group ring.  

The following is a rephrasing in more general terms of Theorem~1 in  \cite{Hurley1}.  Specifically, in \cite{Hurley1}, $R$ is assumed to be a field.  The proof is identical and simply consists of showing that addition and multiplication is preserved.  

\begin{thm} \label{theorem:homomorphism} 
Let $R$ be a finite commutative Frobenius ring and let $G$ be a group of order $n$. 
Then the map $\sigma: RG \rightarrow M_n(R)$ is an injective ring homomorphism.  
\end{thm}

For an element $v = \sum \alpha_i g_i \in RG$, define the element $v^T \in RG$ as 
$v^T = \sum \alpha_i g_i^{-1}.$   This is sometimes known as the canonical involution for the group ring. The reason this notation is used in this setting will be apparent by the next lemma.

The following is a straightforward generalization of a result in \cite{Hurley1}. 

\begin{lem} \label{lemma:transpose}
Let $R$ be a finite commutative Frobenius ring and let $G$ be a group of order $n$. 
For an element  $v \in RG$, we have that $\sigma(v)^T = \sigma(v^T).$  
\end{lem}
\begin{proof}
The $ij$-th element of $\sigma(v^T)$ is $\alpha_{(g_i^{-1} g_j)^{-1}} = \alpha _{g_j^{-1} g_i}$ which is the $ji$-th element of $\sigma(v).$ 
\end{proof}

We next give our first result about the structure of our constructed codes.

\begin{lem} \label{lemma:selforthogonal}
Let $R$ be a finite commutative Frobenius ring and let $G$ be a group of order $n$. 
If $v=v^T$ and $v^2=0$ then $C_v$ is  a self-orthogonal code. 
\end{lem}
\begin{proof}
If $v=v^T$ then $\sigma(v)^T = \sigma(v^T)$ by Lemma~\ref{lemma:transpose}.  Then we have that $(\sigma(v)\sigma(v))_{ij} $ is the inner-product of the $i$-th and $j$-th rows of $\sigma(v)$.   Since $v^2=0$, by Theorem~\ref{theorem:homomorphism} we have that $\sigma(v)\sigma(v) = {\bf 0}.$  This gives that any two rows of $\sigma(v)$ are orthogonal and hence they generate a self-orthogonal code.
\end{proof}

We can now use this lemma  to construct self-dual codes.  For codes over fields we could simply use the dimension of $\sigma(v)$, however over an arbitrary Frobenius ring we cannot determine the size of the generated code simply from the rank of the matrix.  Therefore, we have the following theorem.

\begin{thm} \label{theorem:constructionsd}
Let $R$ be a finite commutative Frobenius ring and let $G$ be a group of order $n$ and let $v \in RG$.
If $v=v^T$, $v^2=0$ and $|C_v| = |R|^{\frac{n}{2}}$ then $C_v$ is a self-dual code.
\end{thm}
\begin{proof}
By Lemma~\ref{lemma:selforthogonal} the code $C_v$ is self-orthogonal and since $|C_v| = |R|^{\frac{n}{2}}$ we have that $C_v$ is self-dual.
\end{proof}

Notice that unlike the field case we are not assuming that $n$ is even.  For example, let  $R=R_k$ and let $G$ be the trivial group of size 1 and let $v= u_i e_G$ where $e_G$ is the identity of the group.  Then $\sigma(v) = ( u_i)$ and $C_v$ is a self-dual code of length 1. 

In the following example, we show the strength of this construction by constructing a code over $R_1$ using the alternating group on 4 letters which has an image under the associated Gray map of the length 24 extended Golay code.  

\begin{exam} \label{example:Golayring} 
We shall use the previous results to construct the binary Golay code from the ring $R_1$.
Let $v=u(b+ab+ac+bc^2)+(bc+bc^2)+(1+u)(c^2+abc^2) \in R_1 A_4$. Then, 
 $C_v $ is a self-dual code of length 12 over $R_1$.  Hence $\phi_k(C)$ is a binary self-dual code of length 12 by Theorem~\ref{theorem:images}. 
 The binary code $\phi_k(C)$ has a generator matrix of the following form:
$\begin{pmatrix} I_{12} & A \end{pmatrix}$ where $A=\left(\begin{smallmatrix}
1&0&1&1&0&0&1&0&1&1&0&1\\
1&1&1&0&0&1&1&0&1&0&1&0\\
1&1&1&1&1&0&0&0&0&1&1&0\\
1&0&1&0&1&0&0&1&1&0&1&1\\
1&0&0&1&1&1&1&0&0&0&1&1\\
1&1&0&0&1&1&0&0&1&1&0&1\\
1&1&0&1&0&1&1&1&0&1&0&0\\
0&1&1&0&1&0&1&1&1&1&0&0\\
0&1&0&1&1&1&0&1&1&0&1&0\\
0&0&1&1&1&1&0&1&0&1&0&1\\
0&1&1&1&0&0&1&1&0&0&1&1\\
0&0&0&0&0&1&1&1&1&1&1&1\\
\end{smallmatrix} \right)$. It is a simple computation to see that  $\phi_k(C_v)$ is the $[24,12,8]$ Golay code.
\end{exam}



\begin{lem} \label{lemma:equivalent}
Let $R$ be a finite commutative Frobenius ring and let $G$ be a group of order $n$. 
If $v = \sum \alpha_i g_i$ and $w = \alpha_i g_ih$ for some $h \in G$ then $C_v$ and $C_w$ are equivalent codes.
\end{lem}
\begin{proof}
The generator matrix for $C_w$ is formed from the generator matrix of $C_v$ by permuting the columns corresponding to multiplication of the elements of $G$ by $h$.  Hence the codes are equivalent.  
\end{proof}


 

\begin{exam} Let $v_1=1+xz+yz+xyz \in \mathbb{F}_2(C_2 \times C_2 \times C_2)$ where $\langle x,y,z \rangle \cong C_2 \times C_2 \times C_2$. Now $\sigma(v_1)$ is equivalent to 
$\left(\begin{smallmatrix}
1&0&0&0&0&1&1&1\\
0&1&0&0&1&0&1&1\\
0&0&1&0&1&1&0&1\\
0&0&0&1&1&1&1&0
\end{smallmatrix} \right)$.  The code $C(v_1)$ is the the $[8,4,4]$ extended Hamming code. Next, let us consider $v_2=(1+xz+yz+xyz)y=y+xz+z+xyz \in \mathbb{F}_2(C_2 \times C_2 \times C_2)$. Then $\sigma(v_2)$ is equivalent to
$\left(\begin{smallmatrix}
1&0&0&1&0&0&1&1\\
0&1&0&1&0&1&0&1\\
0&0&1&1&0&1&1&0\\
0&0&0&0&1&1&1&1
\end{smallmatrix} \right)$. Clearly $C(v_1)$ is equivalent to $C(v_2)$.
\end{exam}

\subsection{Binary Golay Code}  

We shall consider constructions of the $[24,12,8]$ binary Golay code from $\mathbb{F}_2G$. Clearly, the automorphism group of the $[24,12,8]$ code is the Mathieu group $M_{24}$ and the only possible groups are
\[SL(2,3), \,D_{24},\,(C_6\times C_2) \rtimes C_2, C_3 \times D_8, C_2 \times A_4 \;\text{and}\;C_2^2 \times D_6\footnote{These groups are SmallGroup(24,$i$)
for $i \in \{3,6,8,10,12,13,14\}$ according to the GAP system \cite{GAP}.}.\]

Initially, it was shown in \cite{BLM} that the $[24,12,8]$ could be constructed from ideals in the group algebra $\mathbb{F}_2S_4$ where $S_4$ is the symmetric group on $4$ elements. In \cite{IT}, the $[24,12,8]$ code was constructed from $\mathbb{F}_2D_{24}$. We shall now separately consider the remaining cases.
\begin{itemize}
\item The group $C_3 \times D_8$\\
Let
\[
v=\sum_{i=1}^4[a^{i-1}(\alpha_{i}+\alpha_{i+4}z+\alpha_{i+8}z^2)+ba^{i-1}(\alpha_{i+12}+\alpha_{i+16}z+\alpha_{i+20}z^2)] \in {\FF}_{2}(C_3 \times D_8)
\]
where $\langle z \rangle =C_3$, $\langle a,b \rangle =D_8$ and $\alpha_i \in \mathbb{F}_2$. Now
\[
\sigma(v)=
\begin{pmatrix}
A&B\\B&A
\end{pmatrix}
\]

where $A=\begin{pmatrix}A_1&A_2&A_3\\A_3&A_1&A_2\\A_2&A_3&A_1  \end{pmatrix}$, $B=\begin{pmatrix}B_1&B_2&B_3\\B_3&B_1&B_2\\B_2&B_3&B_1 \end{pmatrix}$,
\begin{eqnarray*}
A_1&=&cir(\alpha_{1},\alpha_{2},\alpha_{3},\alpha_{4}),\\
A_2&=&cir(\alpha_{5},\alpha_{6},\alpha_{7},\alpha_{8}),\\
A_3&=&cir(\alpha_{9},\alpha_{10},\alpha_{11},\alpha_{12}),\\
B_1&=&rcir(\alpha_{13},\alpha_{14},\alpha_{15},\alpha_{16}),\\
B_2&=&rcir(\alpha_{17},\alpha_{18},\alpha_{19},\alpha_{20}), \\
B_3&=&rcir(\alpha_{21},\alpha_{22},\alpha_{23},\alpha_{24}) \end{eqnarray*} 
and $cir(\alpha_{1},\alpha_{2},\ldots,\alpha_{n})$, $rcir(\alpha_{1},\alpha_{2},\ldots,\alpha_{n})$ are circulant and reverse circulant matrices  respectively and $\alpha_{1},\alpha_{2},\ldots,\alpha_{n}$ is the first row of the respective matrices.
Clearly $\langle \sigma(v) \rangle$ is self dual if $\sigma(v)^T=\sigma(v)$. Now, $\sigma(v)^T=\sigma(v)$ if and only if $a_2=a_4$, $a_5=a_9$, $a_6=a_{12}$, $a_7=a_{11}$, $a_8=a_{10}$, $a_{17}=a_{21}$, $a_{18}=a_{22}$, $a_{19}=a_{23}$ and $a_{20}=a_{24}$. Next, consider elements of $\mathbb{F}_2(C_3 \times D_8)$ of the form
\[ \begin{split}
\left\{ \right.&\alpha_1+\alpha_2(a+a^3)+\alpha_3a^2+\alpha_4(z+z^2)+\alpha_5az(1+a^2z)+\alpha_6a^2z(1+z)+\alpha_7az(a^2+z)\\
&+\sum_{i=1}^{4}b(\alpha_{i+7}+\alpha_{i+11}(z+z^2))a^{i-1} \,|\, \alpha_i \in \mathbb{F}_2 \left. \right\}
\end{split}  \]
\noindent and in particular the element $v_1=1+b[(\hat{a}+1)+(1+a)(\hat{z}+1)] $ of this set where $\hat{a}=\sum_{i=0}^3a^i$ and $\hat{z}=\sum_{i=0}^2z^i$. The matrix $\sigma(v_1)$ is equivalent to
\[ \begin{pmatrix} I & A \\A &I \end{pmatrix} \]
where
\[A=\left(\begin{smallmatrix}0&1&1&1&1&1&0&0&1&1&0&0\\
1&1&1&0&1&0&0&1&1&0&0&1\\
1&1&0&1&0&0&1&1&0&0&1&1\\
1&0&1&1&0&1&1&0&0&1&1&0\\
1&1&0&0&0&1&1&1&1&1&0&0\\
1&0&0&1&1&1&1&0&1&0&0&1\\
0&0&1&1&1&1&0&1&0&0&1&1\\
0&1&1&0&1&0&1&1&0&1&1&0\\
1&1&0&0&1&1&0&0&0&1&1&1\\
1&0&0&1&1&0&0&1&1&1&1&0\\
0&0&1&1&0&0&1&1&1&1&0&1\\
0&1&1&0&0&1&1&0&1&0&1&1  \end{smallmatrix}\right).
\]
It is a small computation to see that $C(v_1)$ is the $[24, 12, 8]$ code. Moreover, it can be shown that the above set contains $128$ elements that
generate the $[24, 12, 8]$ code.

\item The group $C_2 \times A_4$\\
Let
\[
\begin{split}v&= \sum_{i=1}^3(\alpha_{4i-3}+\alpha_{4i-2}a+\alpha_{4i-1}b+\alpha_{4i}ab
+\alpha_{4i+9}x+\alpha_{4i+10}xa+\alpha_{4i+11}xb+\alpha_{4i+21}xab)c^{i-1} \\
&\in {\FF}_{2}(C_2 \times A_4)
\end{split}\]
where $\langle x \rangle =C_2$, $a=(1,2)(3,4)$, $b=(1,3)(2,4)$ and $c=(1,2,3)$ and $\alpha_i \in \mathbb{F}_2$. Now
\[
\sigma(v)=
\begin{pmatrix}
A&B\\B&A
\end{pmatrix}
\]
where $A=\begin{pmatrix}A_2&A_2&A_3\\A_4&A_5&A_6\\A_7&A_8&A_9  \end{pmatrix}$, $B=\begin{pmatrix}B_2&B_2&B_3\\B_4&B_5&B_6\\B_7&B_8&B_9 \end{pmatrix}$, \newline
$A_1=bc(\alpha_{1},\alpha_{2},\alpha_{3},\alpha_{4})$,
$A_2=bc(\alpha_{5},\alpha_{6},\alpha_{7},\alpha_{8})$, 
$A_3=bc(\alpha_{9},\alpha_{10},\alpha_{11},\alpha_{12})$, \newline
$A_4=bc(\alpha_{9},\alpha_{12},\alpha_{10},\alpha_{11})$,
$A_5=bc(\alpha_{1},\alpha_{4},\alpha_{2},\alpha_{3})$,
$A_6=bc(\alpha_{5},\alpha_{8},\alpha_{6},\alpha_{7})$,\newline
$A_7=bc(\alpha_{5},\alpha_{7},\alpha_{8},\alpha_{6})$,
$A_8=bc(\alpha_{9},\alpha_{11},\alpha_{12},\alpha_{10})$, 
$A_9=bc(\alpha_{1},\alpha_{3},\alpha_{4},\alpha_{2})$,\newline
$B_1=bc(\alpha_{13},\alpha_{14},\alpha_{15},\alpha_{16})$,
$B_2=bc(\alpha_{17},\alpha_{18},\alpha_{19},\alpha_{20})$,
$B_3=bc(\alpha_{21},\alpha_{22},\alpha_{23},\alpha_{24})$,\newline
$B_4=bc(\alpha_{21},\alpha_{24},\alpha_{22},\alpha_{23})$,
$B_5=bc(\alpha_{13},\alpha_{16},\alpha_{14},\alpha_{15})$,
$B_6=bc(\alpha_{17},\alpha_{20},\alpha_{18},\alpha_{19})$,\newline
$B_7=bc(\alpha_{17},\alpha_{19},\alpha_{20},\alpha_{18})$,
$B_8=bc(\alpha_{21},\alpha_{23},\alpha_{24},\alpha_{22})$ and
$B_9=bc(\alpha_{13},\alpha_{15},\alpha_{16},\alpha_{14})$ \newline where $bc(a,b,c,d)$ is a matrix that takes the form
$\left( \begin{smallmatrix}a&b&c&d\\b&a&d&c\\c&d&a&b\\d&c&b&a  \end{smallmatrix} \right)$. Now, $\sigma(v)=\sigma(v)^{T}$ if and only if
$a_5=a_9$, $a_6=a_{12}$, $a_7=a_{10}$, $a_8=a_{11}$, $a_{17}=a_{21}$, $a_{18}=a_{24}$, $a_{19}=a_{24}$ and $a_{20}=a_{23}$. Next, consider elements of $\mathbb{F}_2(C_2 \times A_4)$ of the form
\[ \begin{split}
\{  &\sum_{i=0}^1x^i(  (\alpha_{8i+1}+\alpha_{8i+2}a+\alpha_{8i+3}b+\alpha_{8i+4}ab)
+ 
 \\
&  (\alpha_{8i+5}+\alpha_{8i+6}a+\alpha_{8i+7}b+\alpha_{8i+8}ab)(c+c^2)  )  
 \,|\, \alpha_i \in \mathbb{F}_2  \},
\end{split}  \]
 and in particular the element $v_1=1+x(1+b(1+a)(1+c^2))+xa(1+b)c$ of this set.  The matrix $\sigma(v_1)$ is equivalent to
\[ \begin{pmatrix} I & A \\A &I \end{pmatrix} \]
where
\[A=\left(\begin{smallmatrix}1&0&1&1&0&1&0&1&0&0&1&1\\
0&1&1&1&1&0&1&0&0&0&1&1\\
1&1&1&0&0&1&0&1&1&1&0&0\\
1&1&0&1&1&0&1&0&1&1&0&0\\
0&1&0&1&1&1&0&1&0&1&1&0\\
1&0&1&0&1&1&1&0&1&0&0&1\\
0&1&0&1&0&1&1&1&1&0&0&1\\
1&0&1&0&1&0&1&1&0&1&1&0\\
0&0&1&1&0&1&1&0&1&1&1&0\\
0&0&1&1&1&0&0&1&1&1&0&1\\
1&1&0&0&1&0&0&1&1&0&1&1\\
1&1&0&0&0&1&1&0&0&1&1&1  \end{smallmatrix}\right).
\]
It is a small computation to see that $C(v_1)$ is the $[24, 12, 8]$ code. Moreover, it can be shown that the above set contains $384$ elements that
generate the $[24, 12, 8]$ code.

\item The group $G=(C_6 \times C_2) \rtimes C_2$\\
Let
\[
\begin{split}v&= \sum_{i=1}^4(\alpha_{i}y^{i-1}+\alpha_{i+4}xy^{i-1}+\alpha_{i+8}x^2y^{i-1}+\alpha_{i+12}y^{i-1}z
+\alpha_{i+16}xy^{i-1}z+\alpha_{i+20}x^2y^{i-1}z)\\
& \in {\FF}_{2}((C_6 \times C_2) \rtimes C_2)
\end{split}\]
where $(C_6 \times C_2) \rtimes C_2=\langle x,y,z \,|\,x^3=y^4=z^2=1,\,xy=yx^{2},\,xz=zx,\,yz=zy^3 \rangle$ and $\alpha_i \in \mathbb{F}_2$. Now,
\[
\sigma(v)=\left(
\begin{smallmatrix}
\G_{1}&\G_{2}&\G_{3}&\G_{4}&\G_{5}&\G_{6}&\G_{7}&\G_{8}&\G_{9}&\G_{10}&\G_{11}&\G_{12}&\G_{13}&\G_{14}&\G_{15}&\G_{16}&\G_{17}&\G_{18}&\G_{19}&\G_{20}&\G_{21}&\G_{22}&\G_{23}&\G_{24}\\
\G_{2}&\G_{1}&\G_{13}&\G_{10}&\G_{14}&\G_{12}&\G_{9}&\G_{18}&\G_{7}&\G_{4}&\G_{24}&\G_{6}&\G_{3}&\G_{5}&\G_{17}&\G_{22}&\G_{15}&\G_{8}&\G_{21}&\G_{23}&\G_{19}&\G_{16}&\G_{20}&\G_{11}\\
\G_{3}&\G_{13}&\G_{1}&\G_{14}&\G_{10}&\G_{16}&\G_{17}&\G_{23}&\G_{15}&\G_{5}&\G_{21}&\G_{22}&\G_{2}&\G_{4}&\G_{9}&\G_{6}&\G_{7}&\G_{20}&\G_{24}&\G_{18}&\G_{11}&\G_{12}&\G_{8}&\G_{19}\\
\G_{14}&\G_{10}&\G_{4}&\G_{1}&\G_{2}&\G_{7}&\G_{16}&\G_{24}&\G_{12}&\G_{13}&\G_{18}&\G_{15}&\G_{5}&\G_{3}&\G_{22}&\G_{17}&\G_{6}&\G_{21}&\G_{8}&\G_{11}&\G_{20}&\G_{9}&\G_{19}&\G_{23}\\
\G_{5}&\G_{4}&\G_{10}&\G_{2}&\G_{1}&\G_{9}&\G_{22}&\G_{11}&\G_{6}&\G_{3}&\G_{8}&\G_{17}&\G_{14}&\G_{13}&\G_{16}&\G_{15}&\G_{12}&\G_{19}&\G_{18}&\G_{24}&\G_{23}&\G_{7}&\G_{21}&\G_{20}\\
\G_{24}&\G_{11}&\G_{19}&\G_{17}&\G_{9}&\G_{1}&\G_{8}&\G_{4}&\G_{18}&\G_{15}&\G_{22}&\G_{2}&\G_{21}&\G_{7}&\G_{20}&\G_{3}&\G_{23}&\G_{10}&\G_{6}&\G_{5}&\G_{12}&\G_{13}&\G_{14}&\G_{16}\\
\G_{17}&\G_{9}&\G_{7}&\G_{19}&\G_{21}&\G_{23}&\G_{1}&\G_{6}&\G_{13}&\G_{11}&\G_{5}&\G_{18}&\G_{15}&\G_{24}&\G_{2}&\G_{8}&\G_{3}&\G_{22}&\G_{14}&\G_{12}&\G_{10}&\G_{20}&\G_{16}&\G_{4}\\
\G_{23}&\G_{18}&\G_{8}&\G_{6}&\G_{12}&\G_{14}&\G_{24}&\G_{1}&\G_{21}&\G_{22}&\G_{9}&\G_{10}&\G_{20}&\G_{16}&\G_{11}&\G_{4}&\G_{19}&\G_{13}&\G_{7}&\G_{2}&\G_{15}&\G_{5}&\G_{3}&\G_{17}\\
\G_{9}&\G_{17}&\G_{15}&\G_{11}&\G_{24}&\G_{18}&\G_{13}&\G_{22}&\G_{1}&\G_{19}&\G_{4}&\G_{23}&\G_{7}&\G_{21}&\G_{3}&\G_{20}&\G_{2}&\G_{6}&\G_{10}&\G_{16}&\G_{14}&\G_{8}&\G_{12}&\G_{5}\\
\G_{10}&\G_{14}&\G_{5}&\G_{13}&\G_{3}&\G_{15}&\G_{12}&\G_{21}&\G_{16}&\G_{1}&\G_{23}&\G_{7}&\G_{4}&\G_{2}&\G_{6}&\G_{9}&\G_{22}&\G_{24}&\G_{20}&\G_{19}&\G_{8}&\G_{17}&\G_{11}&\G_{18}\\
\G_{12}&\G_{6}&\G_{22}&\G_{18}&\G_{23}&\G_{21}&\G_{5}&\G_{9}&\G_{14}&\G_{8}&\G_{1}&\G_{19}&\G_{16}&\G_{20}&\G_{4}&\G_{11}&\G_{10}&\G_{7}&\G_{13}&\G_{17}&\G_{3}&\G_{24}&\G_{15}&\G_{2}\\
\G_{11}&\G_{24}&\G_{21}&\G_{15}&\G_{7}&\G_{2}&\G_{18}&\G_{10}&\G_{8}&\G_{17}&\G_{16}&\G_{1}&\G_{19}&\G_{9}&\G_{23}&\G_{13}&\G_{20}&\G_{4}&\G_{12}&\G_{14}&\G_{6}&\G_{3}&\G_{5}&\G_{22}\\
\G_{13}&\G_{3}&\G_{2}&\G_{5}&\G_{4}&\G_{22}&\G_{15}&\G_{20}&\G_{17}&\G_{14}&\G_{19}&\G_{16}&\G_{1}&\G_{10}&\G_{7}&\G_{12}&\G_{9}&\G_{23}&\G_{11}&\G_{8}&\G_{24}&\G_{6}&\G_{18}&\G_{21}\\
\G_{4}&\G_{5}&\G_{14}&\G_{3}&\G_{13}&\G_{17}&\G_{6}&\G_{19}&\G_{22}&\G_{2}&\G_{20}&\G_{9}&\G_{10}&\G_{1}&\G_{12}&\G_{7}&\G_{16}&\G_{11}&\G_{23}&\G_{21}&\G_{18}&\G_{15}&\G_{24}&\G_{8}\\
\G_{15}&\G_{7}&\G_{9}&\G_{21}&\G_{19}&\G_{20}&\G_{2}&\G_{12}&\G_{3}&\G_{24}&\G_{14}&\G_{8}&\G_{17}&\G_{11}&\G_{1}&\G_{18}&\G_{13}&\G_{16}&\G_{5}&\G_{6}&\G_{4}&\G_{23}&\G_{22}&\G_{10}\\
\G_{19}&\G_{21}&\G_{24}&\G_{7}&\G_{15}&\G_{3}&\G_{23}&\G_{14}&\G_{20}&\G_{9}&\G_{12}&\G_{13}&\G_{11}&\G_{17}&\G_{18}&\G_{1}&\G_{8}&\G_{5}&\G_{16}&\G_{10}&\G_{22}&\G_{2}&\G_{4}&\G_{6}\\
\G_{7}&\G_{15}&\G_{17}&\G_{24}&\G_{11}&\G_{8}&\G_{3}&\G_{16}&\G_{2}&\G_{21}&\G_{10}&\G_{20}&\G_{9}&\G_{19}&\G_{13}&\G_{23}&\G_{1}&\G_{12}&\G_{4}&\G_{22}&\G_{5}&\G_{18}&\G_{6}&\G_{14}\\
\G_{18}&\G_{23}&\G_{20}&\G_{22}&\G_{16}&\G_{10}&\G_{21}&\G_{13}&\G_{24}&\G_{6}&\G_{17}&\G_{14}&\G_{8}&\G_{12}&\G_{19}&\G_{5}&\G_{11}&\G_{1}&\G_{15}&\G_{3}&\G_{7}&\G_{4}&\G_{2}&\G_{9}\\
\G_{16}&\G_{22}&\G_{6}&\G_{23}&\G_{18}&\G_{24}&\G_{4}&\G_{17}&\G_{10}&\G_{20}&\G_{13}&\G_{11}&\G_{12}&\G_{8}&\G_{5}&\G_{19}&\G_{14}&\G_{15}&\G_{1}&\G_{9}&\G_{2}&\G_{21}&\G_{7}&\G_{3}\\
\G_{20}&\G_{8}&\G_{18}&\G_{12}&\G_{6}&\G_{5}&\G_{11}&\G_{2}&\G_{19}&\G_{16}&\G_{7}&\G_{4}&\G_{23}&\G_{22}&\G_{24}&\G_{10}&\G_{21}&\G_{3}&\G_{9}&\G_{1}&\G_{17}&\G_{14}&\G_{13}&\G_{15}\\
\G_{22}&\G_{16}&\G_{12}&\G_{20}&\G_{8}&\G_{11}&\G_{10}&\G_{15}&\G_{4}&\G_{23}&\G_{3}&\G_{24}&\G_{6}&\G_{18}&\G_{14}&\G_{21}&\G_{5}&\G_{17}&\G_{2}&\G_{7}&\G_{1}&\G_{19}&\G_{9}&\G_{13}\\
\G_{21}&\G_{19}&\G_{11}&\G_{9}&\G_{17}&\G_{13}&\G_{20}&\G_{5}&\G_{23}&\G_{7}&\G_{6}&\G_{3}&\G_{24}&\G_{15}&\G_{8}&\G_{2}&\G_{18}&\G_{14}&\G_{22}&\G_{4}&\G_{16}&\G_{1}&\G_{10}&\G_{12}\\
\G_{8}&\G_{20}&\G_{23}&\G_{16}&\G_{22}&\G_{4}&\G_{19}&\G_{3}&\G_{11}&\G_{12}&\G_{15}&\G_{5}&\G_{18}&\G_{6}&\G_{21}&\G_{14}&\G_{24}&\G_{2}&\G_{17}&\G_{13}&\G_{9}&\G_{10}&\G_{1}&\G_{7}\\
\G_{6}&\G_{12}&\G_{16}&\G_{8}&\G_{20}&\G_{19}&\G_{14}&\G_{7}&\G_{5}&\G_{18}&\G_{2}&\G_{21}&\G_{22}&\G_{23}&\G_{10}&\G_{24}&\G_{4}&\G_{9}&\G_{3}&\G_{15}&\G_{13}&\G_{11}&\G_{17}&\G_{1}
\end{smallmatrix} \right)
\]

and $\sigma(v)=\sigma(v)^{T}$ if and only if $a_4=a_{14}$, $a_6=a_{24}$, $a_7=a_{17}$, $a_8=a_{23}$, $a_{11}=a_{12}$, $a_{16}=a_{19}$ and $a_{21}=a_{22}$. Next, consider elements of $\mathbb{F}_2((C_6 \times C_2) \rtimes C_2)$ of the form
\[ \begin{split}
&\{\sum_{i=1}^4
(\alpha_{i}y^{i-1}+\alpha_{i+4}xy^{i-1})+\sum_{i=1}^2(\alpha_{i+8}x^2y^{i-1}+\alpha_{i+12}y^{i+1}z)+(\alpha_{11}x^2y^2+\alpha_{17}x^2z)(1+y) \\
&\quad +\alpha_{4}yz+\alpha_{6}x^2y^3z+\alpha_{7}xz+x^2y^2z\alpha_{8}+\alpha_{12}z+\alpha_{14}xy^2z+\alpha_{15}xyz+\alpha_{16}xy^3z\}
\end{split}  \]
\noindent and in particular the element $v_1= 1+[a+b+b^3+(a+a^2)(b^2+b^3)]c$ of this set.  The matrix $\sigma(v_1)$ is equivalent to
\[ \begin{pmatrix} I & A \end{pmatrix} \]
where
\[A=\left(\begin{smallmatrix}0&1&0&1&1&0&1&1&0&0&1&1\\
0&1&0&0&1&1&1&0&1&1&1&0\\
1&0&1&0&1&0&1&1&1&0&1&0\\
0&0&1&1&1&1&1&0&1&0&0&1\\
1&1&0&1&1&1&0&0&0&1&0&1\\
1&1&1&0&0&1&1&0&0&0&1&1\\
1&1&0&0&0&0&0&1&1&1&1&1\\
1&0&1&1&0&1&0&0&1&1&1&0\\
1&1&0&1&0&1&1&1&1&0&0&0\\
0&0&1&0&1&1&0&1&0&1&1&1\\
1&0&1&1&0&0&1&1&0&1&0&1\\
0&1&1&1&1&0&0&1&1&1&0&0\\  \end{smallmatrix}\right).
\]
It is a small computation to see that $C(v_1)$ is the $[24, 12, 8]$ code. Moreover, it can be shown that the above set contains $576$ elements that
generate the $[24, 12, 8]$ code.

\item The group $SL(2,3)$\\
Let
\[
\begin{split}v&= \sum_{i=1}^6x^{i-1}\left(\alpha_i+\alpha_{6+i}y+
\alpha_{12+i}y^2+
\alpha_{18+i}y^2x\right)
\in \FF_2 SL(2,3)
\end{split}\]
where $SL(2,3)=\langle  x,y \,|\, x^3=y^3=(xy)^2\rangle$ and $\alpha_i \in \mathbb{F}_2$. Now,
\[
\sigma(v)=
\begin{pmatrix}
A_1&A_2&A_3&A_4\\
A_5&A_6&A_7&A_8\\
A_9&A_{10}&A_{11}&A_{12}\\
A_{13}&A_{14}&A_{15}&A_{16}
\end{pmatrix},
\]
where
$A_1=circ(\alpha_1,\alpha_2,\alpha_3,\alpha_4,\alpha_5,\alpha_6)$,
$A_2=circ(\alpha_7,\alpha_8,\alpha_9,\alpha_{10},\alpha_{11},\alpha_{12})$,
\newline
$A_3=circ(\alpha_{13},\alpha_{14},\alpha_{15},\alpha_{16},\alpha_{17},\alpha_{18})$,
$A_4=circ(\alpha_{19},\alpha_{20},\alpha_{21},\alpha_{22},\alpha_{23},\alpha_{24})$,
\newline
$A_5=circ(\alpha_{16},\alpha_{22},\alpha_8,\alpha_{13},\alpha_{19},\alpha_{11})$,
$A_6=circ(\alpha_1,\alpha_{21},\alpha_{14},\alpha_4,\alpha_{24},\alpha_{17})$,
\newline
$A_7=circ(\alpha_7,\alpha_{20},\alpha_5,\alpha_{10},\alpha_{23},\alpha_2)$,
$A_8=circ(\alpha_{18},\alpha_{12},\alpha_6,\alpha_{15},\alpha_9,\alpha_3)$,
\newline
$A_9=circ(\alpha_{10},\alpha_{15},\alpha_{21},\alpha_7,\alpha_{18},\alpha_{24})$,
$A_{10}=circ(\alpha_{16},\alpha_6,\alpha_{20},\alpha_{13},\alpha_3,\alpha_{23})$,
\newline
$A_{11}=circ(\alpha_1,\alpha_{12},\alpha_{19},\alpha_4,\alpha_9,\alpha_{22})$,
$A_{12}=circ(\alpha_2,\alpha_{17},\alpha_{11},\alpha_5,\alpha_{14},\alpha_8)$,
\newline
$A_{13}=circ(\alpha_9,\alpha_{14},\alpha_{20},\alpha_{12},\alpha_{17},\alpha_{23})$,
$A_{14}=circ(\alpha_{15},\alpha_5,\alpha_{19},\alpha_{18},\alpha_2,\alpha_{22})$,
\newline
$A_{15}=circ(\alpha_6,\alpha_{11},\alpha_{24},\alpha_3,\alpha_8,\alpha_{21})$,
$A_{16}=circ(\alpha_1,\alpha_{16},\alpha_{10},\alpha_4,\alpha_{13},\alpha_7)$.\\

Now, $\sigma(v)=\sigma(v)^{T}$ if and only if $\alpha_2=\alpha_6$, $\alpha_3=\alpha_5$,
$\alpha_7=\alpha_{16}$, $\alpha_8=\alpha_{11}$, $\alpha_9=\alpha_{19}$, $\alpha_{10}=\alpha_{13}$, $\alpha_{12}=\alpha_{22}$,
$\alpha_{14}=\alpha_{24}$, $\alpha_{15}=\alpha_{18}$, $\alpha_{17}=\alpha_{21}$ and $\alpha_{20}=\alpha_{23}$. Next, consider elements of $\mathbb{F}_2 SL(2,3)$ of the form:
\[ \begin{split}
\{  &
\alpha_1+\alpha_2(x+x^5)+\alpha_3(x^2+x^4)+\alpha_4x^3+\alpha_5(y+x^3 y^2)+
\alpha_6(x y+x^4 y)+\alpha_7(x^2 y+y^2 x) \\
& +
\alpha_{8}(x^3 y+y^2)+
\alpha_{9}(x^5 y+x^3 y^2 x)+
\alpha_{10}(x y^2+x^5 y^2 x)+\alpha_{11}(x^2 y^2+x^5 y^2) \\
&+
\alpha_{12}(x^4 y^2+x^2 y^2 x)+
\alpha_{13}(x y^2 x+x^4 y^2 x) \,|\, \alpha_i \in \mathbb{F}_2 \}.
\end{split}  \]

It can be shown that it is not possible to construct the $[24,12,8]$ from any element of this set.\\

\item The group $C_2^2 \times D_6$

Let $v=$ 
\[
\sum_{i=0}^2[(\alpha_{i+1}+\alpha_{i+4}z+\alpha_{i+7}w+\alpha_{i+10}zw)+b(\alpha_{i+13}+\alpha_{i+16}z+\alpha_{i+19}w+ 
\alpha_{i+22}zw)   ]a^{i} \in {\FF}_{2}(C_2^2 \times D_6)
\]
where $\langle z,w \rangle =C_2^2$, $\langle a,b \rangle =D_6$ and $\alpha_i \in \mathbb{F}_2$. Now
\[
\sigma(v)=
\begin{pmatrix}
A&B\\B&A
\end{pmatrix}
\]

where $A=\begin{pmatrix}A_1&A_2&A_3&A_4\\A_2&A_1&A_4&A_3\\A_3&A_4&A_1&A_2\\A_4&A_3&A_2&A_1  \end{pmatrix}$, $B=\begin{pmatrix}B_1&B_2&B_3&B_4\\B_2&B_1&B_4&B_3\\B_3&B_4&B_1&B_2\\B_4&B_3&B_2&B_1  \end{pmatrix}$,
$A_1=cir(\alpha_{1},\alpha_{2},\alpha_{3})$,
$A_2=cir(\alpha_{4},\alpha_{5},\alpha_{6})$,
$A_3=cir(\alpha_{7},\alpha_{8},\alpha_{9})$,
$A_4=cir(\alpha_{10},\alpha_{11},\alpha_{12})$,
$B_1=rcir(\alpha_{13},\alpha_{14},\alpha_{15})$,
$B_2=rcir(\alpha_{16},\alpha_{17},\alpha_{18})$,
$B_3=rcir(\alpha_{19},\alpha_{20},\alpha_{21})$  and
$B_4=rcir(\alpha_{22},\alpha_{23},\alpha_{24})$.\\
Now, $\sigma(v)=\sigma(v)^{T}$ if and only if $\alpha_2=\alpha_3$, $\alpha_5=\alpha_6$,
$\alpha_8=\alpha_{9}$ and $\alpha_{11}=\alpha_{12}$. Next, consider elements of $\mathbb{F}_2 (C_2^2 \times D_6)$ of the form
\[ \begin{split}
\{&\alpha_{1}+\alpha_{3}z+\alpha_{5}w+\alpha_{7}zw+(a+a^2)(\alpha_{2}+\alpha_{4}z+\alpha_{6}w+\alpha_{8}zw)\\
&+\sum_{i=0}^2+ba^{i}(\alpha_{i+13}+\alpha_{i+16}z+\alpha_{i+19}w+\alpha_{i+22}zw)    \}.
\end{split}  \]

It can be shown that it is not possible to construct the $[24,12,8]$ Golay code  from any element of this set.

We summarize these results in the following theorem.  

\begin{thm}
The $[24, 12, 8]$ Type II code can be constructed in $\mathbb{F}_2 G$ precisely for the
following groups of order $24$: $S_4$, $D_{24}$, $C_3 \times D_8$, $C_2 \times A_4$ and $(C_6\times C_2) \rtimes C_2$.
\end{thm}

\end{itemize}

\section{The Dihedral Group}

Let $D_{2k}$ be the dihedral group.  We describe the group by $D_{2k} = \langle a,b \ | \ a^2 =  b^k =1, ab=b^{-1}a \rangle.$  
The ordering of the elements for the map $\sigma$ is 
$1,b,b^2,\dots, b^{k-1}, a, ab,ab^2,\dots, ab^{k-1}.$  It is this group that McLoughlin used in \cite{McLoughlin} to give a construction of the binary $[48,24,12]$ extremal Type~II code.

Let $v = \sum \alpha_{a^i,b^j} a^i b^j$.  In this case, the matrix  $\sigma(v)$ is of the form:
\begin{equation}
\left(
\begin{array}{cccccccccc}
\alpha_{1} & \alpha_{b} & \alpha_{b^2} & \dots & \alpha_{b^{k-1}} & \alpha_{a} & \alpha_{ab} & \alpha_{ab^2} & \dots & \alpha_{ab^{k-1}} \\
\alpha_{b^{k-1}} & \alpha_{1} & \alpha_{b} & \dots  & \alpha_{b^{k-2}} & \alpha_{ab} & \alpha_{ab^2} & \alpha_{ab^3} & \dots & \alpha_{a} \\
\vdots & \vdots & \vdots & \vdots & \vdots & \vdots & \vdots & \vdots & \vdots & \vdots \\
\alpha_{b} & \alpha_{b^2} & \alpha_{b^3} & \dots & \alpha_{1} & \alpha_{ab^{k-1}} & \alpha_{a} & \alpha_{ab} & \dots & \alpha_{ab^{k-2}} \\
\alpha_{a} & \alpha_{ab} & \alpha_{ab^2} & \dots & \alpha_{ab^{k-1}} & \alpha_{1} & \alpha_{b} & \alpha_{b^2} & \dots & \alpha_{b^{k-1}} \\
\alpha_{ab} & \alpha_{ab^2} & \alpha_{ab^3} & \dots & \alpha_{a} & \alpha_{b^{k-1}} & \alpha_{1} & \alpha_{b} & \dots & \alpha_{b^{k-2}} \\
\vdots & \vdots & \vdots & \vdots & \vdots & \vdots & \vdots & \vdots & \vdots & \vdots \\
\alpha_{ab^{k-1}} & \alpha_{a} & \alpha_{ab} & \dots & \alpha_{ab^{k-2}} & \alpha_{b} & \alpha_{b^2} & \alpha_{b^3} & \dots & \alpha_{1} \\
\end{array} 
\right).
\end{equation}

This gives that $\sigma(v)$ is of the form:
$$
\left( \begin{array}{cc}
A & B  \\ 
B&A  \\ 
\end{array}
\right)
$$
where $A$ is a circulant matrix and $B$ is a reverse circulant matrix.

We begin by   proving  a lemma.

\begin{lem}\label{lemma:char2}
 Let $R$ be  a finite commutative Frobenius ring of characteristic 2.  Let $C$ be the code generated by a matrix $M$ of the form
$$
\left( \begin{array}{cc}
 I_k & B \\
B & I_k  \\
\end{array}
\right),
$$
where $B$ is a symmetric $k$ by $k$ matrix.  If the free rank of $C$ is $k$ then $C$ is self-dual.
\end{lem}
\begin{proof}
Let $D = \langle (I_k | B) \rangle $ and $D'= \langle (B | I_k) \rangle.$   
 The inner-product of  the $i$-th row of $ (I_k | B)$ and the $j$-th row of $(B|I_k)$ is $B_{i,j}  + B_{j,i} = 0$ since $B_{i,j} = B_{j,i}$ and the characteristic is 2. Therefore $D' = D^\perp$ since $|D||D'| = |R|^n.$
 
 The code $C = \langle D, D^\perp \rangle.$  If $D \neq D^\perp$ then $|C| > |D|.$  However, we are assuming that the  free rank of $C$ is $k$.  Hence $C= D = D^\perp.$  This gives that $C$ is a self-dual code.    
\end{proof}

In \cite{Hurley1}, Hurley  proves that $C_v$ is self-dual over $\FF_2$ if $v \in \FF_2D_{24}$,  $v^2=0$ and the dimension is $\frac{n}{2}$.  We can expand this by showing the following which eliminates the need for $v$ to satisfy $v^2=0.$  

\begin{thm}\label{theorem:cyclicconstruction2}
Let $R$ be a finite commutative Frobenius ring of characteristic 2  and let $v \in RD_n$ with $v = \sum \alpha_i h_i$ where only one 
$\alpha_{a^0 b^i}$ is 1 and the rest are 0. 
  If $C_v$ has free  rank $k$, then $C_v$ is  a self-dual code.  
\end{thm}
\begin{proof}
Since only one $\alpha_{2i}$ is 1 and the rest are 0, the generator matrix of $C_v$  is permutation equivalent to a matrix of the form:
$$
\left( \begin{array}{cc} I_k& B  \\ 
B&I_k  \\ 
\end{array}
\right)
$$
where $B$ is a reverse circulant matrix and hence symmetric.
Then by Lemma~\ref{lemma:char2} we have the result.
\end{proof} 

To show the importance of the strengthening of this result, consider the element $v=1 +ab \in \FF_2D_{2k}$ where $k$ is greater than 2.  Then $(1e_{D_{2k}} +ab)^2 \neq 0$ but $C_v$ is a self-dual code. We continue with a larger example.

\begin{exam}
Consider $v \in \FF_2D_{48}$ such that $dim (C_v) = 24$ and the minimum distance of $C_v$ is 10.  There are  192 elements $v$ which produce  equivalent self-dual codes using the technique.  For more information about the importance of this result, see \cite{Open}.
\end{exam}

A common technique for producing self-dual codes is to generate a code with the matrix $(I_{\frac{n}{2}} |  A)$ where $A$ is a reverse circulant matrix.  
Given a code $C$ generated by this matrix we have that $C^\perp$ is generated by $(A^T |  I_{\frac{n}{2}})$ which is equal to $ (A |  I_{\frac{n}{2}}) $ since $A$ is symmetric.  If $C$ is a self-dual code then $\langle (A |  I_{\frac{n}{2}})  \rangle \subseteq \langle  (I_{\frac{n}{2}} |  A) \rangle $. This means that the code generated by 
$\left( \begin{array}{cc} I_{\frac{n}{2}}  & A \\ A & I_{\frac{n}{2}}  \end{array} \right)$ is the code $C$. Consider the first row of this matrix.  Reading this as an element $v \in \FF_2D_{2k}$ we have that $C = C(v).$  This gives the following.

\begin{thm}\label{theorem:putative}
Let $C$ be a binary self-dual code generated by $(I_{\frac{n}{2}} | A)$ where $A$ is a reverse circulant matrix then $C = C(v)$ for some $v \in \FF_2D_{2k}.$ 
\end{thm}

Applying Corollary~\ref{corollary:72}, we have the following.

\begin{cor}
The putative $[72,36,16]$ Type~II code cannot be produced by $(I_{\frac{n}{2}} | A)$ where $A$ is a reverse circulant matrix.
\end{cor}
\begin{proof}
Corollary~\label{corollary:72} gives that the  $[72,36,16]$ Type~II code is not formed from an element in a group algebra and so by Theorem~\ref{theorem:putative} gives the result.  
\end{proof}

This corollary eliminates a commonly used technique in the attempt to construct this putative code. This give a reason why these attempts have not been successful. 

\section{The Cyclic Group cross the Dihedral Group}

In this section, we shall use the group $G= C_s \times D_{2k}$.  Let $C_s = \langle h \rangle$ and let 
$D_{2k} = \langle a,b \ | \ a^2 = b^k =1, ab=b^{-1}a \rangle.$ 
We shall order the elements as follows:
\begin{eqnarray*} 
&\{ (1,1),(1,b),\dots,(1,b^{k-1}), (h,1),(h,b),\dots,(h,b^{k-1}), \dots, (h^{s-1} ,1), \\ &(h^{s-1},b), \dots,(h^{s-1},b^{k-1}),
  (1,ab),\dots,(1,ab^{k-1}), (h,1),(h,ab),  
  \dots,(h,ab^{k-1}),  \\&\dots, (h^{s-1} ,1),(h^{s-1},ab),\dots,(h^{s-1},ab^{k-1}) \}.
  \end{eqnarray*}
  
We see that 
if we choose $v \in RG$ such that only 1 of $\alpha_{(h^i,a^0b^j)}$ is 1 and the rest are 0. Then we get a matrix $\sigma(v)$ of the form:
$$
\left(
\begin{array}{cc} I_k & B \\
B & I_k
\end{array}
\right),
$$  
where $B$ is of the following form:
$$B = 
\left(
\begin{array}{ccccc}
1 A & h A & h^2 A & \dots & h^{s-1} A \\
h^{s-1} A & 1 A & h A &  \dots & h^{s-2} A \\
\vdots & \vdots & \vdots & \vdots & \vdots \\
hA & h^2A &h^3 A & \dots & 1A\\
\end{array}
\right)
$$
where $h^kA$ indicates the matrix where the $i,j$-th element is $(h^k,A_{i,j})$ and $A$ is a reverse circulant matrix. 

\begin{thm}\label{theorem:cyclicconstruction2}
Let $R$ be a Frobenius ring and let $v \in RC_sD_{2k}$ with $v = \sum \alpha_i h_i$ 
where only 1 of $\alpha_{(h^i,)a^0b^j}$ is 1 and the rest are 0.
 Let  $R$ be a finite commutative  Frobenius ring of characteristic 2.  If $|C_v| = |R|^{\frac{n}{2}}$, then $C_v$ is isodual and hence formally self-dual with respect to any weight enumerator.  
 \end{thm}
\begin{proof}
We have that code $C(v)$ is generated by $(I_k | B)$ and then its orthogonal is generated by $(B^T | I_k).$  Then
we have that $B$ is equivalent to $B^T$.   Therefore $C(v)$ and $C(v)^\perp$ are equivalent and therefore formally self-dual with respect to any weight enumerator.
\end{proof}

Note that if $R$ is a finite field, then the condition in the previous theorem becomes that $dim (C_v) = \frac{n}{2}.$

\begin{exam}  Let G be the group $C_3D_8$. There are exactly $2^{12} = 4096$ elements in $\mathbb{F}_2 G$ with
the property that $\alpha_{(h^i,a^0b^j)}$ is equal to $1$ when $i=j=0$ and equal to $0$ otherwise.
 Of these $256$ have $dim(C_v) = 12$
and $192$ of these codes are formally self-dual but not self-dual and $64$ are self-dual. Of the
$192$ formally self-dual codes $80$ have minimum distance $6$ which is optimal for Type I codes.
As an example, if $v_1 = 1+a(b+b(1+b)(bh+h^2))$
then $C_{v_1}$ is a formally self-dual code
with minimum distance $6$. The remaining $112$ formally self-dual codes have have minimum distance $4$ and $C_{v_2}$ is an example
of such a code where $v_2=1+a(b^2+h+b^3h+h^2+bh^2)$.
\end{exam}

\begin{exam} Let G be the group $C_4D_8$ and consider elements of $\mathbb{F}_2 G$ with
the property that $\alpha_{(h^i,a^0b^j)}$ is equal to $1$ when $i=j=0$ and equal to $0$ otherwise. Of these elements, there are $2048$ that  have $dim(C_v) = 16$, of these $512$  are self-dual and the remaining $1536$ are formally self-dual. Let
$v_1=1+a(\hat{b}+h)h$, $v_2=1+a(b+b^3+h+h^3+(b^2+\hat{b})h^2+(1+\hat{b})h^3)$ and $v_3=1+a(b(1+h)+ \hat{b}h^2+(b+\hat{b})h^3)$.  The code $C_{v_1}$ is an example of a 
formally self-dual with minimum distance $4$, the code $C_{v_2}$ is an example of a
formally self-dual with minimum distance $6$ and the code $C_{v_3}$ is an example of a
formally self-dual with minimum distance $8$. Of the $1536$ formally self-dual codes, there are $896$ with minimum distance $4$, $192$ with minimum distance $6$ and $448$ with minimum distance $8$.
\end{exam}

\section{Cyclic Case} 

In this section, we shall set $G =C_n$ the cyclic group of order $n$.
Since the inception of cyclic codes, it has been an open question to determine which cyclic codes were self-dual.  We shall describe when this occurs. 

We focus  on the case when $n=2k$.  Let $G = \langle h \rangle.$ 
Then let $h_i = h^i$.  We then use as the ordering of the elements of $G$:
$$(h_0,h_2,\dots,h_{2k}, h_1,h_3,\dots,h_{2k-1}).$$
That is $g_i = h_{2(i-1)}$ for $i=1$ to $k$ and 
$g_{k+j} = h_{2(j-1)+1}$ for $j=1$ to $k$.

It follows that the form of $\sigma(v)$ is:

\[
\begin{pmatrix}
\alpha_{h_{0}} & \alpha_{h_{2}} & \cdots & \alpha_{h_{2k}} & \alpha_{h_{1}} & \alpha_{h_{3}} & \cdots & \alpha_{h_{2k-1}} \\
\alpha_{h_{2k}} & \alpha_{h_{0}} & \cdots & \alpha_{h_{2k-2}} & \alpha_{h_{2k-1}} & \alpha_{h_{1}} & \cdots & \alpha_{h_{2k-3}} \\
\vdots & \vdots & \ddots & \vdots & \vdots & \vdots & \ddots & \vdots\\ 
\alpha_{h_{4}} & \alpha_{h_{6}} & \cdots & \alpha_{h_{2}} & \alpha_{h_{3}} & \alpha_{h_{5}} & \cdots & \alpha_{h_{1}} \\
\alpha_{h_{2k-1}} & \alpha_{h_{1}} & \cdots & \alpha_{h_{2k-3}} & \alpha_{h_{0}} & \alpha_{h_{2}} & \cdots & \alpha_{h_{2k}} \\
\alpha_{h_{2k-3}} & \alpha_{h_{2k-1}} & \cdots & \alpha_{h_{2k-5}} & \alpha_{h_{2k}} & \alpha_{h_{0}} & \cdots & \alpha_{h_{2k-2}} \\
\vdots & \vdots & \ddots & \vdots & \vdots & \vdots & \ddots & \vdots\\
\alpha_{h_{1}} & \alpha_{h_{3}} & \cdots & \alpha_{h_{2k-1}} & \alpha_{h_{4}} & \alpha_{h_{6}} & \cdots & \alpha_{h_{2}} 
\end{pmatrix}.
\]

Hence $\sigma(v)$ is of the form
\[ \begin{pmatrix}A&B\\D&A  \end{pmatrix}\]

\noindent where $A$, $B$ and $D$ are circulant matrices.

Choose an element of $v$ such that $v = \sum \alpha_i h_i$ where only one of $\alpha_{2i} =1$ and the rest of $\alpha_{2i}$ are 0.  Then the generating matrix is permutation equivalent to a matrix where  $A$ is $I_k$ and $B$ and $D$ are  circulant matrices.  Namely, we get a matrix of the form 
$$\left( \begin{array}{cc}
I_{\frac{n}{2}} & B \\
D & I_{\frac{n}{2}}
\end{array} \right).$$

\begin{thm}\label{theorem:cyclicconstruction}
Let $R$ be a Frobenius ring of characteristic 2 and let $v\in RC_n$ with $v = \sum \alpha_i h_i$ where only one $\alpha_{2i} =1$ and the rest of $\alpha_{2i}$ are 0. If $v_{2k-i} = v_i$ for odd $i$ and $|C| = |R|^k$ then $C(v)$ is a self-dual code. 
\end{thm} 
\begin{proof}
  By the construction we have that $\sigma(v)$ is of the form
$$
\left( \begin{array}{cc}
I_k & B \\
D& I_k \\
\end{array}
\right).
$$
If $v_{2k-i} = v_i$ for odd $i$ then $D=B^T$.  
We have that $|C| = |R|^k.$  However, the form of the matrix gives that $C$ contains a free code isomorphic to $R^k$, namely the code generated by the matrix $(I_k | B).$ This means that $C = \langle (I_k |  B)\rangle.$

Consider the code generated by the matrix $( B^T |  I_k)$.  This code must be $C^\perp.$
However, this code is contained in $C(v)$ as well, so we have that $C=C^\perp.$ 
\end{proof}

Notice that we did not have to determine the cardinality of the code to see that the code was self-dual.  

Note that it is certainly more difficult to use this technique to construct self-dual codes with the cyclic group.  That is, we had to put more restrictions on $v$ to obtain a self-dual code.  This is certainly to be expected since it is fairly difficult to find cyclic self-dual codes.

Moreover, note that a code over $R_k$ constructed here is cyclic, which gives that its image under the Gray map is quasi-cyclic of index $2^k$. 

\begin{exam}
Let $G$ be the cyclic group of order 10 and $v = 1 +uh + h^5 +uh^9 \in R_1C_{10}$.  Then
$C_v = \langle \sigma(v),u\sigma(v) \rangle$ is cyclic self-dual code and its image under $\phi_1$ is a  binary quasi-cyclic self-dual $[20,10,4]$  code of index 2.
\end{exam}

We note that this is a standard construction of self-dual codes, namely you take a vector $\vv$ and generate a circulant matrix $B$ from it with $BB^T = - I_k$, with $n=2k$,  and generate the code $(I_k | B).$  Hence we have another of the standard constructions of self-dual codes within our general framework.

 We can now use our general construction to produce isodual codes.  
 
 \begin{thm} \label{theorem:cyclicfsd}
 Let $R$ be a finite commutative Frobenius ring with characteristic 2. Let $v\in RC_n$ with $v = \sum \alpha_i h_i$ where only one $\alpha_{2i} =1$ and the rest of $\alpha_{2i}$ are 0.  If $|C(v)| = |R|^{\frac{n}{2}}$ then $C(v)$ is a formally self-dual code with respect to any weight enumerator.  
 \end{thm}
 \begin{proof}
 If $|C(v)| = |R|^{\frac{n}{2}}$ then $C$ is generated by the matrix $(I_k |  B)$ where $B$ is a circulant matrix.   Then its orthogonal is of the form $(B^T |  I_k).$  Since $B$ is a circulant code, then by permuting the rows and columns of $B$ we can form $B^T$.  This gives that $C(v)^\perp$ is equivalent to $C(V)$ and hence isodual and therefore formally self-dual code with respect to any weight enumerator. 
 \end{proof}

\begin{exam}
Let $G$ be the cyclic group of order 6 and 
$v= 1+u_2h + (1+u_1+u_1u_2) h^3 + u_1h^5 \in R_2 C_6.$  
Then $C_v = \langle \sigma(v), u_1 \sigma(v), u_1u_2 \sigma(v) \rangle $ is a cyclic formally self-dual code and its image under $\phi_2$  is a binary quasi-cyclic self-dual $[24,12,6]$ code of index 4.
\end{exam}


\begin{thebibliography}{99.}%
\bibitem{BLM}
Bernhardt, F., Landrock, P., Manz, O., The extended Golay codes considered as ideals, J. Combin. Theory Ser. A, {\bf 55}, 1990, no. 2., 235 - 246.

 
\bibitem{Bouyuklieva}   Bouyuklieva,  S., On the automorphism of order $2$ with fixed points for the extremal self-dual codes of length $24m$, {\em Des. Codes Cryptogr.}, {\bf 25},  2002, no.1, 5 - 13.

\bibitem{Boyuklieva-2} Bouyuklieva,   S.,   O'Brien, E.A.  Willems, W.,   The automorphism group of a binary self-dual doubly-even $[72,36,16]$ code is solvable, {\em IEEE Trans. Inform. Theory}, 52,  2006, 4244 - 4248.


\bibitem{Dougherty0} Dougherty, S.T., Kaya, A., Salutrk, E., Constructions of Self-Dual Codes and Formally Self-Dual Codes over Rings, AAECC, DOI 10.1007/s00s00-016-0288-5, 2016.

\bibitem{Open} Dougherty, S.T., Kim, J.L., Sol\'e, P., Open Problems in Coding Theory, Contemporary Mathematics, Volume 634, 2015, 79 - 99.



\bibitem{Dougherty5} Dougherty, S.T., Yildiz, B., Karadeniz, S., Codes over $R_k$, Gray maps and their Binary Images, with  Finite Fields and their Applications, {\bf 17}, no. 3, 2011,  205 - 219.

\bibitem{Dougherty6} Dougherty, S.T., Yildiz, B., Karadeniz, S., Cyclic Codes over $R_k$,  Designs, Codes and Cryptography, {\bf 63}, no. 1., 2012,  113 - 126. 

\bibitem{Dougherty8} Dougherty, S.T., Yildiz, B., Karadeniz, S., Self-dual codes over $R_k$ and binary self-dual codes. Eur. J. Pure Appl. Math. {\bf 6}, no. 1, 2013, 89 - 106. 


\bibitem{GAP}{The GAP Group}, {GAP --  Groups, Algorithms and Programming}, {Version 4.4},
 {2006 (http:/www.gap-system.org)}.

\bibitem{Karadeniz1} Karadeniz, S.; Dougherty, S. T., Yildiz, B., Constructing formally self-dual codes over $R_k$, Discrete Appl. Math. {\bf 167},  2014, 188 - 196. 

\bibitem{Hurley1}  Hurley, T., Group Rings and Rings of Matrices, Int. Jour. Pure and Appl. Math,, {\bf 31}, no. 3, 2006, 319 - 335.

\bibitem{Lam} Lam, C. W. H.; Thiel, L.; Swiercz, S. The nonexistence of finite projective planes of order 10. Canad. J. Math. {\bf 41},   no. 6, 1989,  1117 - 1123. 


\bibitem{McLoughlin} McLoughlin, I., A group ring construction of the $[48,24,12]$ Type~II linear block code, Des. Codes Cryptogr., {\bf 63}, no. 1, 
2012, 
29 - 41.  

\bibitem{IT} McLoughlin, I., Hurley, T., {A group ring construction of the extended binary Golay code},
{IEEE Trans. Inform. Theory},
   {\bf 54}, 2008, no. 9, 4381 - 4383.


\bibitem{Nebe} Nebe, G., An extremal $[72,36,16]$ binary code has no automorphism group containing $\ZZ_2\times \ZZ_4, Q_8$ or $\ZZ_{10}$, {\em Finite Fields Appl.}, {\bf 18}, 2012, no. 3, 563 - 566.





\bibitem{OBrien} O'Brien, E.A., Willems, W.,  On the automorphism group of a binary self-dual doubly-even $[72,36,16]$ code, {\em IEEE Trans. Inform. Theory}, { \bf 57}, no.7, 2011, 4445 - 4451.

\bibitem{Yankov}  Yankov, N.,  A putative doubly-even $[72,36,16]$ code does not have an automorphism of order $9$, {\em IEEE Trans. Inform. Theory}, {\bf 58}, no. 1, 2012, 159 - 163.

\end{thebibliography}
\end{document}